\documentclass{amsart}
\usepackage[all]{xy}
\usepackage{amsmath,amssymb,dsfont,mathtools,enumitem}
\usepackage[pdftitle={Categorical interpretations of some key agreement protocols},
  pdfauthor={Nick Inassaridze, Tamaz Kandelaki, Manuel Ladra},
  pdfsubject={Applied Mathematics},
  pdfkeywords={cryptography},
]{hyperref}
\usepackage[lite]{amsrefs}
\usepackage{microtype}

\newtheorem{thm}{Theorem}[section]

\theoremstyle{definition}

\theoremstyle{remark}
\newtheorem{rem}[thm]{Remark}
\numberwithin{equation}{section}

\newcommand*{\Hom}{\operatorname{Hom}}
\newcommand*{\N}{\mathbb N}

\def\kap{KAP}
\def\ckap{CKAP}
\def\eckap{ECKAP}

\begin{document}
\title{Categorical interpretations of some key agreement protocols}
\author{Nick Inassaridze}
\address{Departamento de Matem\'atica Aplicada I, Universidad de Vigo, EUIT
Forestal, 36005 Pontevedra, Spain}
\address{Tbilisi Centre for Mathematical Sciences, Chavchavadze Ave.~75, 3/35, Tbilisi 0168, Georgia}
\address{A.~Razmadze Mathematical Institute of Tbilisi State University, University Street~2, Tbilisi 0143, Georgia}
\email{niko.inas@gmail.com}

\author{Tamaz Kandelaki}
\address{Tbilisi Centre for Mathematical Sciences, Chavchavadze Ave.~75, 3/35, Tbilisi 0168, Georgia}
\address{A.~Razmadze Mathematical Institute of Tbilisi State University, University Street~2, Tbilisi 0143, Georgia}
\email{kandel@rmi.acnet.ge}
\email{tam.kandel@gmail.com}

\author{Manuel Ladra}
\address{Departamento de \'Algebra, Facultad de Matem\'aticas, Universidad de Santiago de Compostela, 15782 Santiago de Compostela, Spain}
\email{manuel.ladra@gmail.com}

\thanks{The first and third authors were supported by Ministerio de Ciencia e Innovaci\'on (European FEDER support included), grant MTM2009-14464-C02, and by Xunta de Galicia, grant Incite09 207 215 PR}

\begin{abstract}
We give interpretations of some known key agreement protocols in the framework of category theory and in this way we give a method of constructing of many new key agreement protocols.
\end{abstract}


\maketitle

\

\section{Introduction}

Key agreement is one of the fundamental cryptographic primitives after encryption and digital signature. Key agreement protocols (\kap s) allow two or more parties to exchange information among themselves over an adversarially controlled insecure network and agree upon a common session key, which may be used for later secure communication among the parties. Thus, secure \kap s serve as basic building block for constructing secure, complex, higher level cryptographic protocols.

The first pioneering work for key agreement is the Diffie-Hellman protocol given in their seminal paper \cite{DiHe} that invents the public key cryptography and revolutionizes the field of modern cryptography. In \cite{DiHe} a two-party key agreement protocol was proposed. There have been many attempts to provide authentic key agreement based on the Diffie-Hellman protocol \cite{DiOoWi, MaTaIm, MeQuVa, Ya}.

In the last few years some efforts have been made to construct \kap\ using hard problems in infinite non-commutative groups. Here we only mention the idea based on conjugacy search problem which were reckoned as potentially hard problem for construction of one-way functions \cite{AnAnGo, KoLe}. To realize proposed algorithms the main attempts were directed to the suitable platform group selection.

Recently in \cite{SaLiTva} the \kap\ has been constructed using matrix power functions based on matrix ring action on some matrix set and generalizing the Diffie-Hellman \kap. It has been suggested that main advantage of the proposed \kap\ is considerable fast computations and avoidance of arithmetic operations with long integers.

The aim of this work is to suggest a general scheme of constructing \kap s using the category theory. We assume the reader is familiar with categories (we refer to the classical book of Mac Lane \cite{MLa} for the background in Category Theory). Based on the structure of categories, we present the above mentioned \kap s as very particular cases of our categorical \kap s. Working new examples of our categorical \kap s will be given in subsequent papers.

\;

\section{Key Agreement Protocols Related to Categories}

In this section we define \kap  s which are arisen from the structure of categories.

\subsection{\kap\ based on categories}\label{subsec:KAPC}
 Let ${\mathcal C}$ be a (non-empty) category and let $A,B$ be objects of $\mathcal C$ such that $\Hom(A,B)\neq \varnothing$. We suggest the set $\Hom(A,B)$ to be a set of possible keys, while $\Hom(A,A)$ and $\Hom(B,B)$ are monoids which can be used by Alice and Bob, respectively, for actions on $\Hom(A,B)$ if they wish to create a shared key. According to the structure of the category $\mathcal C$, Alice is able to act on the set of possible keys using the right action of $\Hom(A,A)$ on $\Hom(A,B)$. Similarly, Bob is able to act on the set of possible keys using the left action of $\Hom(B,B)$ on $\Hom(A,B)$. Let $g$ be a publicly known element of the set $\Hom(A,B)$. Then, for creating a shared key, Alice and Bob can proceed as follows:
\begin{enumerate}
\item[1.]  Alice selects at random an element $f\in \Hom(A,A)$ and computes composition $g\cdot f$, and sends it to Bob;
\item[2.] Bob selects at random an element $h\in \Hom(B,B)$ and computes composition $h\cdot g$, and sends it to Alice;
\item[3.] Alice computes $k_a=(h\cdot g)\cdot f$, while Bob computes $k_b = h\cdot (g\cdot f)$;
\item[4.] Since $(h\cdot g)\cdot f = h\cdot (g\cdot f)$, the shared key is $k=k_a = k_b\in \Hom(A,B)$.
\end{enumerate}
This protocol, based on the structure of the category $\mathcal C$, is called the categorical key agreement protocol (\ckap).

\;

\subsection{General form \kap\  based on enriched categories}\label{subsec:KAPEC}

In this subsection we give another scheme of \kap\  induced by a structure of a category, but which is enriched over the category of abelian groups, i.e. a category whose morphism sets are abelian groups satisfying some axioms (see \cite{MLa}). This construction generalizes the \kap\ given in previous subsection and motivated by some known \kap s. Namely, our approach makes it possible to interpret many known \kap s as particular cases of our construction.

Let ${\mathcal D}$ be a (non-empty) enriched category over the category of abelian groups.  Clearly, it means that for any objects $A$ and $B$ in this category  $\Hom(A,A)$ and $\Hom(B,B)$ are unital rings, $\Hom(A,B)$ is an abelian group and composition of morphisms in $\mathcal D$ is bilinear. Let $A,B$ be objects of $\mathcal D$ such that $\Hom(A,B)\neq \varnothing$.
Let $m,\;n\in \N$ be natural numbers, $\mathcal{A}_A$ and $\mathcal{B}_A$ commuting subrings of the $n\times n$-matrix ring $M_n\big(\Hom(A,A)\big)$, while $\mathcal{A}_B$ and $\mathcal{B}_B$ commuting subrings of $m\times m$-matrix ring $M_m\big(\Hom(B,B)\big)$.
Let $\varphi$ be a publicly known $m\times n$-matrix over the abelian group $\Hom(A,B)$. If Alice and Bob wish to create a common secret key, they can proceed as follows:
\begin{enumerate}
\item[1.] Alice selects at random matrices $\psi _a \in \mathcal{A}_A$ and $\omega _a\in \mathcal{A}_B$, computes product of matrices $\omega _a\cdot \varphi \cdot \psi _a$, and sends it to Bob;
\item[2.] Bob selects at random matrices $\psi _b \in \mathcal{B}_A$ and $\omega _b\in \mathcal{B}_B$, computes product of matrices $\omega _b\cdot \varphi \cdot \psi _b$, and sends it to Alice;
\item[3.] Alice computes $k_a=\omega _a\cdot \omega _b\cdot \varphi \cdot \psi _b\cdot \psi _a$, while Bob computes $k_b=\omega _b\cdot \omega _a\cdot \varphi \cdot \psi _a\cdot \psi _b$;
\item[4.]  Since $\omega _a\cdot \omega _b=\omega _b\cdot \omega _a$ and $\psi _b\cdot \psi _a=\psi _a\cdot \psi _b$, the shared secret key is the $m\times n$-matrix $k=k_a = k_b$ over the abelian group $\Hom(A,B)$.
\end{enumerate}

This protocol is called the enriched categorical key agreement protocol (\eckap).  The following assertion relates two categorical \kap s presented in this section.

\begin{thm}\label{thm:KAPEX}
There is a universal faithful functor $T$ from the category of categories to the category of enriched categories over the category of abelian groups. According to this correspondence,   any \ckap\ related to a category $\mathcal C$ can be interpreted as a \eckap\ related to the  enriched category $T({\mathcal C})$.
\end{thm}

\begin{proof}
We just construct the functor $T$ and omit the proof of its universality since it directly follows from the construction. In fact, for any category $\mathcal C$ define the category $T({\mathcal C})$ as follows: its objects class coincides with the objects class of $\mathcal C$, while $\Hom_{T({\mathcal C})}(A,B)$ is the free abelian group generated by the set $\Hom_{\mathcal C}(A,B)$ for any $A,B\in T({\mathcal C})$. The composition of morphisms in $T({\mathcal C})$ is obviously induced by the composition of morphisms in $\mathcal C$. Then it is easy to check that $\Hom_{T({\mathcal C})}(A,A)$, $A\in T({\mathcal C})$, is unital ring, the composition is bilinear and all axioms of enriched category satisfied. Hence the category $T({\mathcal C})$ is enriched over the category of abelian groups.

Given a category $\mathcal C$, one can obtain its corresponding \ckap\ as \eckap\ of the enriched category $T({\mathcal C})$ by assuming $m=n=1$, and $\mathcal{B}_A$ and $\mathcal{A}_B$ to be subrings of $M_n\big(\Hom(A,A)\big)$ and $M_m\big(\Hom(B,B)\big)$ generated by the unital matrices, respectively.
\end{proof}

\begin{rem}\label{rem:abmon}
In our constructions one can successfully use an enriched category over any symmetric monoidal category, e.g. over the category of abelian monoids (see Theorem \ref{Thm 3.3}).
\end{rem}

\;

\subsection{Security problem of \ckap\ and \eckap\ }

It is assumed that any \kap\ must be secure up to solving a certain mathematical problem in a reasonable length of time. One can see that \ckap\ and \eckap\ are based on the conjecture that a function defined by composition of morphisms in a category is a one-way function in general. We suggest that the security of \ckap\ and \eckap\ depends on concrete model of a given category, i.e. the cardinality of ``Hom-sets'' and non-triviality of the morphism composition. We also would like to mention that the security of our categorical \kap s is not less than the security of the Diffie-Hellman \kap\  \cite{DiHe} and Ko et al. \kap \cite{KoLe}, since they are particular cases of our \kap s (see Section \ref{sec 3}). Further discussion on the security problems will be developed in subsequent papers where the concrete implementations of our \kap s are given.

\;

\section{Interpretations of some well-known \kap s}\label{sec 3}

In this section we show that some of well-known \kap s are particular cases of our general categorical \kap s.

\subsection{Diffie-Hellman \kap\ as \ckap}

Diffie-Hellman Key Agreement Protocol is defined in \cite{DiHe}. It has the following form. Let $G$ be a cyclic group, and $g$ a generator of $G$, where both $g$ and its order $s$ are publicly known. If Alice and Bob wish to create a shared key, they can proceed as follows:
\begin{enumerate}
\item[1.]  Alice selects uniformly at random an integer $m\in [2, s-1]$, computes $g^m$, and sends it to Bob;
\item[2.] Bob selects uniformly at random an integer $n\in [2, s-1]$, computes $g^n$, and sends it to Alice;
\item[3.] Alice computes $k_m = (g^n)^m$, while Bob computes $k_n = (g^m)^n$;
\item[4.] The shared key is thus $k = k_m = k_n\in G$.
\end{enumerate}

\begin{thm}\label{thm:Diffie-Hall KAP is CatKAP}
The Diffie-Hellman  \kap\ is a \ckap\ based on a certainly constructed category. Moreover, one can interpret it as a \eckap.
\end{thm}
\begin{proof}
Let us construct a category $\mathcal C$ as follows: let $\mathcal C$ have only two objects $A$ and $B$; let the morphism sets be $\Hom (A,A)={\mathbb N}$, $\Hom (B,B)={\mathbb N}$,  $\Hom(A,B)=G$ and $\Hom(B,A)=\varnothing$, where ${\mathbb N}$ is the abelian monoid of natural numbers with respect to usual product; and let the composition of morphisms be defined by the formulas
\begin{align*}
n\cdot g=g^n\quad \text{and}\quad g\cdot m=g^m, \;\; m,n\in {\mathbb N},\;g\in G.
\end{align*}
It is easy to see that the \ckap\ arisen from the structure of so defined category $\mathcal C$ is exactly the Diffie-Hellman \kap. Thanks to Theorem \ref{thm:KAPEX}, the rest of the assertion follows.
\end{proof}

\;

\subsection{Ko-Lee-Cheon-Han-Kang-Park \kap\ as \ckap} \label{subsection:Koal-protacol}

Recall Ko-Lee-Cheon-Han-Kang-Park key agreement protocol (Ko et al. \kap) given in \cite{KoLe}. Let $G$ be a non-abelian group and $H_A,\;H_B$ its commuting subgroups. Let $g$ be a publicly known element of $G$. If Alice and Bob wish to create a common secret key, they can proceed as follows:
\begin{enumerate}
\item[1.]  Alice selects at random an element $a\in H_A$, computes ${}^ag = aga^{-1}$, and
sends it to Bob;
\item[2.] Bob selects at random an element $b\in H_B$, computes ${}^bg = bgb^{-1}$, and
sends it to Alice;
\item[3.] Alice computes $k_a = {}^a{({}^bg)}$, while Bob computes $k_b = {}^b{({}^ag)}$;
\item[4.] The common secret key is $k=k_a = k_b\in G$.
\end{enumerate}

\begin{thm}
\label{Ko et al KAP}
Ko et al. \kap\ is a \ckap\ arisen from a certainly constructed category. Moreover, one can interpret it as a \eckap.
\end{thm}
\begin{proof}
Let us construct a category $\mathcal C$ as follows: let $\mathcal C$ have only two objects $A$ and $B$; let the morphism sets be $\Hom (A,A)=H_A$, $\Hom (B,B)=H_B$, $\Hom(A,B)=G$ and $\Hom(B,A)=\varnothing$; and let the composition of morphisms be defined by the equalities
\begin{align*}
a\cdot a'= a'a,\quad b\cdot b'=bb',\quad g\cdot a=aga^{-1}\quad \text{and}\quad b\cdot g=bgb^{-1}
\end{align*}
for $a,a'\in \Hom (A,A)$, $b,b'\in \Hom (B,B)$ and $g\in \Hom (A,B)$. It is clear that the \ckap\ arisen from the category $\mathcal C$ is exactly the Ko et al. \kap. Now, using again Theorem \ref{thm:KAPEX} completes the proof.
\end{proof}

\;

\subsection{Sakalauskas-Listopadskis-Tvarijonas \kap\ as \eckap}

In \cite{SaLiTva} E. Sakalauskas, N. Listopadskis, and P. Tvarijonas defined \kap\ (Sakalauskas et al. \kap) based on matrix power function. Now we recall it but in a slightly reformulated form. Let $\mathcal S$  be a semiring and $\mathcal M$ a $\mathcal S$-semibimodule, i.e. there exist bilinear, right and left actions of $\mathcal S$ on abelian monoid $\mathcal M$ satisfying the following associative law
$$
(lm)r=l(mr)\quad l,r\in {\mathcal S},\; m\in{\mathcal M}.
$$
Let $k$ be a natural number and let $M_k({\mathcal S})$ and $M_k({\mathcal M})$ denote $k\times k$-matrix semiring over $\mathcal S$ and $k\times k$-matrix abelian monoid over $\mathcal M$, respectively. It is well known that $M_k({\mathcal M})$ is a $M_k({\mathcal S})$-semibimodule with respect to the naturally induced right and left actions by the rule of standard matrix multiplication. Let $\varphi$ be a publicly known $k\times k$-matrix in $M_k({\mathcal M})$, while ${\mathcal A}_A$ and ${\mathcal A}_B$ be two subsemirings of commuting matrices in $M_k({\mathcal S})$. If Alice and Bob wish to create a common secret key, they can proceed as follows:
\begin{enumerate}
\item[1.] Alice selects at random secret matrices $\psi_a\in {\mathcal A}_A$ and $\omega_a\in {\mathcal A}_B$, computes product of matrices $\omega_a\cdot \varphi \cdot \psi_a$ and sends it to Bob;
\item[2.] Bob selects at random secret matrices $\psi_b\in {\mathcal A}_A$ and $\omega_b\in {\mathcal A}_B$ computes product of matrices $\omega_b\cdot \varphi \cdot \psi_b$ and sends it to Alice;
\item[3.] Both parties compute the following common secret (key) matrix $k$:
$$
k=\omega_a\cdot\omega_b\cdot \varphi \cdot\psi_b\cdot \psi_a=\omega_b\cdot \omega_a\cdot \varphi \cdot \psi_a\cdot \psi_b.
$$
\end{enumerate}

\begin{thm}\label{Thm 3.3}
Sakalauskas et al. \kap\ is a \eckap\ arisen from a certainly constructed enriched category over the category of abelian monoids.
\end{thm}
\begin{proof}
According to the structure of $\mathcal S$-semibimodule $\mathcal M$, one constructs the enriched category over the category of abelian monoids $\mathcal D$ with two objects $A,\; B$ and the following ``Hom-objects'':
$$
\Hom(A,A)={\mathcal S}\;\;\Hom(B,B)={\mathcal S}\;\;\Hom(A,B)={\mathcal M},\;\;\Hom(B,A)=\varnothing,
$$
while the composition is defined by the right and left actions of $\mathcal S$ on $\mathcal M$. Now, it is obvious that the \eckap\ arisen from the enriched category $\mathcal D$ coincides with the Sakalauskas et al. \kap.
\end{proof}

\;

\section{Categorical multi-party \kap}

This section is only devoted of suggesting a multi-party \kap\ based on the structure of a category, and hence showing an advantage of categorical approach to construct easily multi-party \kap s. Further investigation of our categorical multi-party \kap\ and its working examples will appear in our subsequent papers.

Suppose there is a set $S=\{ A_1,A_2, \ldots , A_n\}$ of $n$ users. If they wish to agree a common secret key and for that to use open insecure channels, they can proceed as follows:

\;

\begin{enumerate}
\item[Step 1.] Chose an order in $S$, i.e.   $S=\big( A_1,A_2, \ldots , A_n \big)$ ;
\item[Step 2.] A category $\mathcal C$ is public. For each user $A_i$, $1\le i\le n$, it is chosen an object $C_i\in {\mathcal C}$ publicly and $(n-1)$ elements $\{g_1, \ldots ,g_{n-1}\}$ such that $g_i\in \Hom(C_i,C_{i+1})$;
\item[Step 3.]  Any user $A_i$, $1\le i\le n$, chose randomly an element $f_i\in \Hom(C_i, C_i)$ and computes
\begin{align*}
& g_if_i\quad  & \text{for} & \quad i=1, &  \\
& f_ig_{i-1}\quad  & \text{for} & \quad i=n, &  \\
& f_ig_{i-1} \quad  \text{and} \quad g_if_i & \text{for} & \quad 1<i<n. &
\end{align*}
Then any user $A_i$ sends $g_if_i$ to any other user $A_j$ for $j>i$ and sends $f_ig_{i-1}$ to any other user $A_j$ for $j<i$.
\item[Step 4.] Thanks to the associative law of morphism composition in the category $\mathcal C$ any user $A_i$ computes
\begin{align*}
k_i & = (f_ng_{n-1})\cdots (f_{i+1}g_i) f_i (g_{i-1}f_{i-1})\cdots (g_1f_1)\\
 & = f_ng_{n-1}\cdots f_{i+1}g_i f_i g_{i-1}f_{i-1}\cdots g_1f_1=k
\end{align*}
and obtain a common element $k\in \Hom(C_1,C_n)$.
\end{enumerate}

\;

\end{document}